\documentclass{llncs} 
\usepackage[title,toc,titletoc,header]{appendix}
\usepackage{amssymb,color}
\usepackage{verbatim}
\usepackage{comment}
\usepackage{xspace}
\usepackage{multirow}
\usepackage{thmtools,thm-restate}
\usepackage{verbatim}
\usepackage{amsmath}
\newcommand{\newcomment}[1]{}
 \allowdisplaybreaks

\newcommand{\R}{\mathbb{R}}

\newcommand{\nnz}{\mathbf{nnz}}
\newcommand{\dist}{\mathrm{dist}}
\newcommand{\tr}{\mathrm{tr}}
\newcommand{\poly}{\mathrm{poly}}
\newcommand\numberthis{\addtocounter{equation}{1}\tag{\theequation}}

\usepackage[boxed,linesnumbered]{algorithm2e}
\SetKwInput{KwInput}{Input}
\SetKwInput{KwOutput}{Output}

\declaretheorem[name=Theorem]{thm}
\newtheorem{fact}[thm]{Fact}
\newtheorem{defi}[thm]{Definition}
\newtheorem{lem}[thm]{Lemma}
\newtheorem{cor}[thm]{Corollary}

\SetKwInOut{Input}{Input}
\SetKwInOut{Output}{Output}
\SetKwInOut{Promise}{Promise}

\begin{document} 
 %\title{Fast and provable seeding for spherical $k$-mean}
 \title{Faster Coreset Construction  for  Projective Clustering \textit{via} Low-Rank Approximation}
\author{
Rameshwar Pratap  \inst{1} \thanks{This work done when author was affiliated with TCS Innovation Labs.}
  \and
Sandeep Sen\inst{2}
}
\institute{Wipro Technologies, Bangalore, India\\
	\email{rameshwar.pratap@gmail.com} 
\and IIT Delhi, India\\
\email{ssen@cse.iitd.ernet.in}} 
   
\date{}
\maketitle

\begin{abstract}

In this work, we  present a randomized coreset construction for projective 
clustering, which  involves computing a set of $k$ closest $j$-dimensional
linear (affine) subspaces of a given set of $n$ vectors in $d$ dimensions. 
Let $A \in \R^{n\times d}$ be an  input matrix. An  earlier deterministic 
coreset construction of Feldman \textit{et. al.}~\cite{FSS13} relied on computing the SVD 
of $A$. The best known algorithms for SVD require $\min\{nd^2, n^2d\}$ time, 
which may not be feasible for large values of $n$ and $d$.
We  present a coreset construction by projecting the matrix $A$ on 
some orthonormal vectors that closely approximate the right singular 
vectors of $A$. As a consequence, 
when the values of  $k$ and $j$ are small,
we are able to achieve a   
faster algorithm,  as compared to~\cite{FSS13}, while maintaining  almost 
the same approximation. We also benefit in terms of space as well as 
exploit the sparsity of the input dataset. Another advantage of our 
approach  is that it can be constructed in a streaming setting 
quite efficiently. 
\end{abstract}

\section{Introduction}\label{sec:Intro} 
\paragraph{Succinct representation of Big data -- Coreset:} 
Recent years have witnessed a dramatic increase in our ability to collect data 
from various sources. % like social media platforms, mobile applications, sensors, finance and biology. 
This data flood has surpassed our ability to understand, 
analyse and process them.  \textit{Big data} is a new terminology that has become
  quite popular in identifying such datasets that are difficult to analyse 
with the current available technologies.  
%In \cite{Gartner15} Gartner has %summarized the major challenges, which come across  while analyzing such Big data sets, 
%in terms of 3V's- volume (the data size); variability (dimensions or different %sources that generates the data); and velocity (the speed at which data arrives).  
%In this work we will be focusing on the first two V's \textit{i.e.} volume and %variability (cardinality and dimension) of Big data analysis aspects. 
One possible approach to manage such large volume %and variety 
of  datasets is to  keep a succinct summary of the datasets such that it approximately preserves 
the required properties of the original datasets. This notion was initially  formalised 
by Agrawal \textit{et al.}~\cite{AgarwalHV04}, 
%in the context of approximating various descriptors of the extent of a  point set.  
% They derived  $\epsilon$-approximation algorithms for computing smallest enclosing ball, computing diameter, width etc., 
and they coined the term \textit{coreset} for such summaries. % of the input. 
Intuitively, a coreset can be considered as a semantic compression of the input. 
For example: in the case of clustering, 
  a coreset is a weighted subset of the data such that  the cost of a clustering algorithm % clustering 
evaluated on the coreset closely approximates to the corresponding cost on the
entire dataset. 
Consider  a   set $Q$ (possibly of infinite size) of query shapes (for example: subspaces,  set of points, 
set of lines etc.), then for every shape $q\in Q$, the sum of distances from $q$ to the  
input points,  and  the sum of distances from $q$ to the points in the  coreset, is approximately the same. 
If the query set belongs to some particular  candidate query set, then such coreset is called as a \textit{weak} 
coreset~\cite{Mahoney11}; and if the coreset approximates the distances from all possible (potentially infinite) 
query shapes, then it is called as \textit{strong} coreset. 
Coresets are a practical  and flexible tool which require  no or  minimal assumption on the data. 
Although the analysis  techniques for coreset construction are a bit involved, and require 
tools from computational geometry and linear algebra, the resulting coreset construction algorithms are  easy 
to implement. Another important property of coresets is that they can be constructed in a streaming 
and distributed setting quite efficiently. This is due to the fact that unions of coresets are coresets, 
and coresets of coresets are also  coresets~\cite{Har-PeledM04}. Also, 
using these properties it is possible to construct coresets in a tree-wise fashion which can be parallelized 
in a Map-Reduce style~\cite{FSS13}.

%In many scenarios, coresets could be a weighted subset of the input point set - weights represent the relative importance of the points.
Coreset constructions have been studied extensively for various data analysis tasks. 
There are usually two steps involved in the coreset construction --
dimensionality reduction,  and cardinality reduction.
 The dimension reduction step of the coreset construction 
 includes projecting points in a low dimension space such that the original 
geometry of points is also preserved  in the low dimension. These projection techniques 
includes SVD decomposition, random projections, row/column subset selections, or any combinations of 
these (see~\cite{FSS13,CohenEMMP151}). The cardinality reduction step includes contracting the input 
size \textit{via} sampling or other geometric analysis approach on  the reduced dimension instance of the input. 
We refer readers to   survey articles of Jeff M. Phillips~\cite{Phillips16} and %an earlier survey by 
Agarwal \textit{et al.}~\cite{AgarwalHV07}. % on geometric coreset construction.% for geometric problems. 

%In the problem of coreset construction the goal is compute a succinct 
%representation of such that it approximately  preserve the required properties of  the entire input. 
% We refer readers to a survey article~\cite{Phillips16}  for  the results of in this area.
 
% for details.% on coreset construction for geometric problems.
In this work,  we focus on the dimension reduction step of coreset construction for the projective clustering problem. 
In the  paragraph below, we discuss the motivation behind the projective clustering problem.

\paragraph{Projective clustering:} Clustering is one of the most popular techniques for analyzing large data, and  is  widely 
used in many areas such as classification, unsupervised   learning, data mining, indexing, pattern recognition. 
Many popular clustering algorithms such as $k$-means, BIRCH~\cite{birch}
, DBSCAN~\cite{dbscan} 
are full 
dimensional -- they give equal importance to all the dimensions while computing the distance between two points. 
These clustering algorithms works well in low dimensional datasets, however, due to the 
  \textit{``curse of dimensionality''}  such algorithms scale poorly in high dimensions. 
Moreover, in   high dimensional datasets a full dimensional distance might not 
be appropriate as farthest neighbour of a point is expected to be roughly as close 
as its nearest neighbour~\cite{HinneburgAK00}. 
These problems are often handled \textit{via} methods such as Principal component analysis (PCA) 
 or Johnson-Lindenstrauss lemma %~\cite{JL83} 
 by finding a low dimensional representation of the data 
  obtained by projecting all points on a subspace 
 so that the information loss is minimized. 
However, projecting all the points in a single low dimensional subspace may not be appropriate 
%faces a problem 
when different clusters lie in different subspaces. This motivates the study of projective clustering   
which involves finding clusters along different subspaces. 
%More precisely,   the problem of  linear (affine) $(k, j)$-projective clustering 
%the goal is to find a set of $k$ closest $j$-dimensional linear (affine) subspaces.
%involves computing a set of $k$ closest $j$-dimensional linear (affine) subspaces 
%of a given set of $n$ points  in $d$ dimensions.
Projective clustering algorithms have been widely applicable for indexing %~\cite{ChakrabartiM00} 
and pattern discovery %~\cite{AggarwalPWYP99} 
in high dimensional datasets.

\subsection{Our contribution}
%In this work, we present  randomized coreset constructions for projective clustering. 
With the above motivation we study the dimension reduction step of coreset construction for projective clustering problem. 
We first briefly describe the subspace and projective clustering problems.  
In a $j$-subspace clustering problem, given  a set of $n,$ $d$ dimensional vectors, 
denoted by $A \in \R^{n\times d}$, 
the problem is to find a  $j$-dimensional subspace such that it minimizes the sum of squared 
distances from the rows of $A$, over every $j$-dimensional subspace. Further, in the problem of  
 linear (affine) $(k, j)$-projective clustering, the goal is to find a  
closed set $\mathcal{C}$  which is the union of $k$ linear (affine) subspaces  
 each of dimension $j$, such that  it minimizes the sum of squared distances from the rows of $A$, over 
 every possible choice of $\mathcal{C}$ (see Definitions~\ref{defi:jsubspace},\ref{defi:kjsubspace}).% for  formal definitions).

 Feldman \textit{et al.}~\cite{FSS13} presented a deterministic 
  coreset  construction for these clustering problems.  Their coreset construction relies on projecting 
the rows of $A$ on the first few right singular values of $A$. However, the main drawback of their construction is that 
it requires computing the SVD of $A$ which is expensive for large values of $n$ and $d.$    Cohen \textit{et al.}\cite{CohenEMMP151} suggested \textit{``projection-cost-preserving-sketch''} for various 
clustering problems. 
Their sketches are essentially the dimensionality reduction step of the coreset construction. Using a low rank approximation 
of $A$, they suggested a faster coreset construction for the 
subspace clustering problem. However, it was not clear that how their techniques could be extended for projective 
clustering problem. % \footnote{private communication}. 
In this work, we extend their techniques and obtain a faster dimension reduction for projective clustering, 
and as a consequence, a faster coreset construction for the projective clustering problem.  
  In Section~\ref{sec:jSubspace}, 
we first revisit the techniques for subspace clustering problem, and in Section~\ref{sec:kjSubspace} we present our coreset 
construction for projective clustering problem. We state our main result as follows:
 (In the following theorem, $\nnz(A)$ denotes the number of non-zero entries of $A.$)
  \begin{theorem}\label{thm:thmmainKJ}
  Let $A\in \R^{n\times d}$, $\epsilon\in(0, 1)$, and $j, k$ be two integers less than $(d-1)$, and $(n-1)$ 
respectively such that $k(j+1) \leq d-1$.
  Then there is a randomized algorithm which outputs a matrix $A^*$ of rank $O\left({k(j+1)}/{\epsilon^2}\right)$
    such that for every    non-empty closed set $\mathcal{C}$, which is the union of $k$ linear 
  (affine) subspaces each of dimension at most $j$,  
    the following holds \textit{w.h.p.}%with high probability:
 $$
   \left|\left(\dist^2(A^*, \mathcal{C})+\Delta^*\right)-\dist^2(A, \mathcal{C})\right|\leq \epsilon\dist^2(A, \mathcal{C}).
$$
 Where, %$m^*=O\left(\frac{k(j+1)}{\epsilon^2}\right)$; 
   $j^*=k(j+1)$; 
   $\Delta^*= ||A-A^{O(\frac{j^*}{\epsilon^2})}||^2_F$;  $\dist^2(A, \mathcal{C})$ denotes the sum 
   of squared distances from each row of $A$ to its closest point in $\mathcal{C}$; and  %$A^{O(\frac{j^*}{\epsilon^2})}$
   $A^{O(\frac{j^*}{\epsilon^2})}$ 
   is the best
   rank $O(\frac{j^*}{\epsilon^2})$ approximation of $A$. The   running time of the algorithm is
     %$O(\nnz(A)( \frac{m^*}{\epsilon}+ m^* \log m^* )+ (n + d)(\frac{m^*}{\epsilon}+m^*\log m^*)^2+ndm^*)$.
       % $O\left(\nnz(A)( \frac{j^*}{\epsilon^3}+ \frac{j^*}{\epsilon^2} \log \frac{j^*}{\epsilon^2} )+(n + d)(\frac{j^*}{\epsilon^3}+\\ \frac{j^*}{\epsilon^2}\log \frac{j^*}{\epsilon^2})^2\right)$. 
  ${O}\left(\nnz(A) \frac{j^*}{\epsilon^3} + (n + d)\frac{{j^*}^2}{\epsilon^6}\right).$%~~\footnote{Here, $\tilde{O}$ is the asymptotic notation that ignores logarithmic factors.}$
    %$\tilde{O}\left(\nnz(A) \frac{k(j+1)}{\epsilon^3} + (n + d)\frac{{(k(j+1))}^2}{\epsilon^6}\right)$. 
        \end{theorem}
\begin{remark}
%We   achieve a faster construction of projective clustering problem as compared to the construction of~\cite{FSS13}. Instead of projecting points on the right singular vectors of $A$, 
We develop our coreset by projecting points 
on some orthonormal vectors that closely approximate the right singular vectors of $A$, 
and we  obtain them using the algorithm of Sarl{\'{o}}s~\cite{sarlos06}.  
% \cite{sarlos06} \footnote{Please note that our result holds in combination 
% of any other low rank approximation algorithms such as~\cite{sarlos06} which provides similar 
% multiplicative approximation guarantee.}. 
The   running time of our algorithm is better than the corresponding deterministic 
algorithm of~\cite{FSS13} 
when $n\geq d$ and $j^*=o(n)$, or, when  $n<d$ and $j^*=o(d)$, where $j^*=k(j+1)$.
%when  $j^*=o(n)$, or  $j^*=o(d)$ (depending whether $n\geq d$, or $n< d$, respectively), where $j^*=k(j+1)$.
Further, as the  coreset construction time depends on the number of non-zero entries of the matrix, our
algorithm is substantially faster for sparse data matrices.
 Please note that one can also use any other low-rank approximation algorithms  
 such as~\cite{ClarksonW13} (instead of~\cite{sarlos06}), which 
 offer multiplicative approximation guarantee.  
 However, for completeness sake we use the bounds of~\cite{sarlos06}, and compare our   results with~\cite{FSS13}.
\end{remark}
\begin{remark}
 The term $\Delta^*$ is a positive constant, and is sum of squared singular values from $O(j^*/\epsilon^2)$ to $d$. 
 We  use $A^*$  to approximately solve the clustering problem, and add the constant $\Delta^*$ in the clustering cost obtained 
 from $A^*$, this sum 
 gives a good approximation \textit{w.r.t.} the cost of  clustering on  $A$.
\end{remark}

\begin{remark}
 An advantage of our coresets is that it can be constructed in the pass efficient 
 streaming model~\cite{HenzingerRR98}, where access to the input is limited to only a constant number 
 of sequential passes. We construct our coreset by projecting 
 the matrix $A$ on orthonormal vectors, 
  that closely approximate the
 %that are good approximation of 
 right singular vectors of $A$, 
 %and due to~\cite{sarlos06}, our algorithm requires only two passes over the data in order to compute those vectors.
  our algorithm requires only two passes over the data in order to compute those orthonormal vectors using~\cite{sarlos06}.
\end{remark}

 \subsection{Related work}
Coreset construction has been studied extensively for the problem of $j$-subspace clustering. However, we will discuss 
a few of them  that are more relevant to our work.
Feldman \textit{et al.}~\cite{FeldmanFS06}  developed a strong coreset whose size is
exponential in $d, j$, logarithmic in $n$, and their coreset construction requires $O(n)$ time.
Feldman \textit{et al.}~\cite{FeldmanMSW10} improved their earlier result~\cite{FeldmanFS06} and developed a  
 coreset of size    logarithmic in $n$, linear in $d$, and exponential in $j$. 
However, the construction requires $O(ndj)$ time. 
In~\cite{FeldmanL11} Feldman and Langberg showed a coreset construction of size  polynomial in $j$ and $d$ (independent of $n$). 
Feldman \textit{et al.}~\cite{FSS13} presented a novel coreset 
construction for subspace and projective clustering.
%they build up their coreset by projecting the points on  a few first right singular vectors of the input matrix where 
They showed that the sum of squared Euclidean distance from 
$n$ rows of $A\in \R^{n \times d}$ to any $j$-dimensional subspace can be approximated upto 
$(1+\epsilon)$ factor, with an additive constant which is the sum of a few last singular values of 
$A$, by projecting the points on the  
first  $O({j}/{\epsilon})$ right singular vectors of   $A$. 
Thus, they were able to show the dimension reduction from $d$ to $O({j}/{\epsilon})$ . 
They also showed $O({k(j+1)}/{\epsilon^2})$ dimension reduction for $(k, j)$-projective clustering problem.
Recently,  for $j$-subspace clustering, Cohen \textit{et al.}~\cite{CohenEMMP151} improved the  construction 
of~\cite{FSS13}   using only first  $\lceil {j}/{\epsilon} \rceil$ right singular vectors, 
which is an improvement over~\cite{FSS13} by a constant factor.
%\textcolor{blue}{include prj clustering dimension reduction.} For $(k, j)$-projective clustering Feldman \textit{et al}~\cite{FSS13} showed that the sum of squared Euclidean distance from $n$ rows of $A\in \R^{n \times d}$ to any $k$ $j$-dimensional subspace can be approximated upto 
%$(1+\epsilon)$ factor, with an additive constant which is the sum of a few last singular values of $A$, by projecting the points on the  first  $O(\frac{k(j+1)}{\epsilon^2})$ right singular vectors of   $A$.  

Sariel Har-Peled~\cite{HarPeled04} showed that for projective clustering problem 
it is not possible to get a strong coreset of size 
sublinear in $n$, even for a simpler instance such as a family of pair of planes in $\R^3$.  However, in a restricted setting, where
points  are on an integer grid, and the largest coordinate of any point is bounded 
by a polynomial  in $n$ and $d$,
Varadarajan \textit {et al.}~\cite{VaradarajanX12} showed that a sublinear sized coreset construction for projective clustering.
% on a restricted setting -- when points  are on an integer grid, and the largest coordinate of any point is bounded 
% by a polynomial  in $n$ and $d$.

\paragraph*{Organization of the paper:}
In Section~\ref{sec:prelim}, we present the necessary notations, definitions and linear 
algebra background that are used in the various proofs in the paper. 
In Section~\ref{sec:jSubspace}, 
we revisit the result of~\cite{CohenEMMP151}, and discuss the coreset construction for subspace 
clustering using their techniques.  In Section~\ref{sec:kjSubspace}, 
we extend the 
result of Section~\ref{sec:jSubspace}, and 
 present the coreset construction for 
projective clustering problem. %Finally in Section~\ref{sec:conclusion}, we conclude our discussion, and state some open questions.
We conclude our discussion, and state some open questions in Section~\ref{sec:conclusion}.
\section{Preliminaries}\label{sec:prelim} 
\vspace{-0.7cm}
\begin{center}
%\begin{tabularx}{\textwidth}{|X|l|}
\begin{tabular}{|c|l|}
\hline
 \multicolumn{2}{|c|}{\bf Notations}\\
 \hline
 $A=U \Sigma V^T$& columns of $U, V$ are orthonormal and called as
 left and  right \\ & singular vectors  of  $A$; $[\Sigma]$ is a diagonal matrix having  the \\ & corresponding singular values\\
%$[A]_{n\times d}=[U]_{n\times n}[\Sigma]_{n\times d}[V^T]_{d\times d}$& columns of $U$ and $V$ are orthonormal, and called as
% left and right singular vectors \\ & of  $A$; $[\Sigma]$ is a diagonal matrix having the corresponding singular values\\
\hline
$A^{(m)}=U\Sigma^{(m)}V^T$ &  $\Sigma^{(m)}$ is the diagonal having the $m$ largest entries of $\Sigma$,  and \\ & $0$ otherwise\\
\hline
%$A=\Sigma_{t=1}^n \sigma_t u^{(t)} {v^{(t)}}^T$ & $u^{(t)}$ and ${v^{(t)}}^T$ are $t$-th orthonormal columns of $U$ and $V$ respectively \\
%\hline
$[X]_{d\times j}$ &  $j$ orthonormal columns represent a $j$-dimensional subspace \\ & $L$ in $\R^d$\\
\hline
$[X^{\perp}]_{d\times (d-j)}$ &$(d-j)$  dimensional subspace $L^{\perp}$ orthogonal to subspace $L$\\ 
%& subspace $L^{\perp}$ (orthogonal to the subspace $L$)\\
\hline
 %$\pi_{V, k}(A)$ &   best rank-$k$ approximation of $A$ with its rows projected on the row span of $V$\\
 %$\pi_{\mathcal{S}}(A)$ &   matrix formed by projecting the rows of $A$   on the row span of $\mathcal{S}$\\
 $\pi_{\mathcal{S}}(A)$ &   matrix formed by projecting   $A$   on the row span of $\mathcal{S}$\\
 \hline
 %$\pi_{\mathcal{S}, k}(A)$ &  the best rank-$k$ approximation of $A$ with its rows projected \\ & on the row span of $\mathcal{S}$\\
 $\pi_{\mathcal{S}, k}(A)$ &  the best rank-$k$ approximation of $A$ after projecting its rows \\ & on the row span of $\mathcal{S}$\\
 \hline
 $A^{(k)}$ & the best rank-$k$ approximation of $A$\\
 \hline
 $\nnz(A)$ & the number of non-zero entries of $A$\\
 \hline
%$V'$ & first $m$ rows of $V^T$\\
%\hline
 \end{tabular}
 %\end{tabularx}
\end{center}
 
 Below we present some necessary linear algebra background. 
We first present some basic properties of 
Frobenius norm of a matrix. We define SVD
(singular value decomposition) of a  matrix, and its basic properties. We describe 
the expression about the distance of a point, and sum of square distances of the rows of matrix -  
from a subspace and a closed set. 

\begin{fact}[Frobenius norm and its properties]\label{fact:forbNormDef}
 Let $A\in \R^{n \times d}$, then  square of  
 Frobenius norm of $A$ is defined as the sum of the absolute squares of its elements, \textit{i.e.} 
 $||A||^2_F=\Sigma_{i=1}^n\Sigma_{j=1}^d a_{i, j}^2$. Further, if $\{\sigma_i\}_{i=1}^d$ are singular
 values of $A$, then $||A||^2_F= \Sigma_{i=1}^d \sigma^2_i$. Also, if $\tr(A)$ be the trace of the matrix $A$
 then $||A||^2_F=\tr(A^TA).$ %If $A^{(m)}$ is the best $m$-rank approximation of $A$, then its holds that
 %$||A-A^{(m)}||^2_F=\Sigma_{i=m+1}^d \sigma^2_i$, and $\tilde{A}$ is any $m$-rank approximation of $A$, then 
 %$||A-\tilde{A}||^2_F\geq\Sigma_{i=m+1}^d \sigma^2_i$,
\end{fact}

\begin{fact}\label{fact:fact3}
   Let  $AX$ be the  projection of points of $A$ on the $j$-dimensional subspace $L$ represented by an orthonormal matrix $X$. We can also write the projection of the points in the rows of $A$ to $L$ as $AXX^T$,  these projected points 
    are still $d$-dimensional, but lie within the $j$-dimensional subspace. Further, $||AX||_F^2=||AXX^T||_F^2.$
\end{fact}

 \paragraph{ The Singular Value Decomposition:}
A matrix $A \in \R^{n\times d}$ of rank at most $r$ can be written due to its SVD decomposition as 
$A=\Sigma_{i=1}^r \sigma_t u^{(i)} {v^{(i)}}^T$. 
Here,  $u^{(i)}$ and ${v^{(i)}}$ are $i$-th orthonormal columns of $U$ and $V$ respectively, 
and $\sigma_1\geq \sigma_2, \ldots \sigma_r\geq 0.$ Also, ${u^{(i)}}^TA=\sigma_i{v^{(i)}}^T $, and $A{v^{(i)}}=\sigma_i{u^{(i)}}$ for $1\leq i\leq r.$ 
Further, 
 the matrix $A^{(k)}$ that minimizes $||A-B||_F$ among all matrices $B$ (of rank at most $k$) is given by 
$  A^{(k)}=\Sigma_{i=1}^k A v^{(i)} {v^{(i)}}^T$- \textit{i.e.} by projecting $A$ on the first $k$ right 
singular vectors of $A$.
%Further, by standard linear algebra, we have $||A_k||^2_F=\Sigma_{i=1}^k \sigma_i^2$, and $||A-A_k||^2_F=\Sigma_{i=k+1}^{n} \sigma_i^2$.
 
\paragraph{$l_2$ distances to a subspace:} Let $L$ be a $j$-dimensional subspace in $\R^d$  represented 
by an orthonormal matrix $X\in \R^{d\times j}$. 
Then,  for a point $p\in \R^d$, $||p^TX||_F^2$ is the  squares of the length of projections of the point $p$ 
 on the subspace $L$. Similarly,   given a matrix $A\in \R^{n\times d}$, $||AX||_F^2$ is the sum of  squares of 
the length  of projections of the points 
(rows) of $A$ on the subspace $L$. Let $L^{\perp}$ be the orthogonal complement of $L$ represented 
by an orthonormal matrix $X^{\perp}\in \R^{d\times (d-j)}$. Then, $||AX^{\perp}||_F^2$ is the sum 
of   squares of distances of the points  of $A$ from  $L$.  

\paragraph{$l_2$ distance  to a closed set:} Let $S\in \R^d$ be a closed set and $p$ be a point in $\R^d$. 
We define the $l_2$ distance between $p$ and $S$ by
$\dist^2(p, S):=\min_{s\in S} \dist^2(p, s)$, \textit{i.e.}, the smallest distance between $p$ and any element $s\in S.$
If $S$ consists of union of $k$,  $j$-dimensional subspaces $L_1, \ldots, L_k$, then $\dist^2(p, S)$ 
denotes the distance from $p$ to the closest set $S$. Similarly,   given a matrix $A\in \R^{n\times d}$,
$\dist^2(A, S):=\Sigma_{i=1}^n\dist^2(A_{i^*}, S).$ Here, $A_{i^*}$ denotes the $i$th row of $A$.

\paragraph{Pythagorean theorem:} Let $A\in \R^{n\times d}$,  $L$ be a $j$-dimensional subspace 
in $\R^d$  represented by an orthonormal matrix $X\in \R^{d\times j}$, and  $L^{\perp}$ be the 
orthogonal complement of the subspace $L$ represented by an orthonormal matrix $X\in \R^{d\times d-j}$. 
Then by Pythagorean theorem we have $||A||_F^2=||AX||^2_F+||AX^{\perp}||^2_F$. 
Further, if $\mathcal{C}$ is a closed set spanned by $X$, then  due to the Pythagorean theorem we have
 $\dist^2(A, \mathcal{C})=||A{{X}^{\perp}}||^2_F+\dist^2(A{X}{X}^T, \mathcal{C})$.
 We will use the following fact in our analysis which hold true due to Pythagorean theorem.
 
 \begin{fact}\label{fact:fact2}
   Let  $A\in \R^{n\times d}$, and $X\in \R^{d\times j}$ be 
a matrix having  first $j$ right singular vectors of $A$ as columns, then %  orthonormal columns, then 
   due to the Pythagorean theorem, we have \\ $||A-AXX^T||_F^2=||A||^2_F-||AXX^T||_F^2.$ 
  % Similarly,  due   to Pythagorean theorem and elementry linear algebra  it can be shown that
 % Further, if ${V'}\in \R^{d\times j}$ be the first $j$ right singular vectors of $A$, then we have
 %  $$(A-A{V'}{V'}^T)^T(A-A{V'}{V'}^T)=A^TA-(A{V'}{V'}^T)^T(A{V'}{V'}^T).$$
\end{fact}

In the following, we state some facts from elementary  linear algebra which are required for 
deriving the correctness of our result.
\begin{fact}\label{fact:trace}
 For a square matrix  $M\in \R^{n \times n}$,  $\tr(M)$ is the sum of all its diagonal entries. 
 Further, for  matrices $A\in \R^{n \times d}, B\in \R^{d \times n}$ due to the cyclic  property of the $\tr$ function, we have 
 $\tr(AB)=\tr(BA)$. Also for square matrices $M, N\in \R^{n \times n}$, due to the linear  property of the $\tr$ function: $\tr(M\pm N)=\tr(M)\pm\tr(N)$.
\end{fact}

\begin{fact}\label{fact:posSemiDef}
 A symmetric matrix $M\in \R^{n \times n}$ is positive semidefinite 
 if $x^TMx>0$ for all $x\in \R^n.$ A matrix $M$ is positive 
 semidefinite then the following two statements are equivalent:
 \begin{itemize}
  \item there is a real nonsingular matrix $N$ such that $M=N^TN$,
  \item  all eigenvalues of $M$ are nonnegative.
 \end{itemize}
\end{fact}

\begin{fact}\label{fact:fact1}
 Let $A\in \R^{n \times d}$ and $U\Sigma V^T$ be the SVD of $A$. Then, the first $j$ columns of $V$ span a subspace that minimizes the sum of squares distances 
of the vectors in $A$ from all $j$-dimensional subspace, 
and this sum is $\Sigma_{i=j+1}^d \sigma_i^2$. Thus, for any $j$-dimensional subspace
represented by an orthonormal matrix
$X$, we have $||AX^{\perp}||^2_F\geq \Sigma_{i=j+1}^d \sigma_i^2. $
\end{fact}

\begin{fact}\label{fact:projDecNorm}
%  If a matrix  $M\in \R^{d\times l}$, and  $X \in \R^{d\times k}$ is an  orthonormal matrix. 
  Let $M\in \R^{d\times l}$ be a matrix.
 Then, for an orthonormal matrix $X \in \R^{d\times k}$, due to  elementary linear 
 algebra we  have,  $||X{X}^TM||_F^2\leq ||M||_F^2.$
\end{fact}
%The Cauchy–Schwarz inequality states that for all vectors x and y of an inner product space it is true that

%\begin{fact}[Cauchy–Schwarz inequality]\label{fact:CS} Let $A\in \R^{p \times q}$ and $B \in \R^{q \times r}$ be the two matrices, then due to 
%Cauchy-Schwarz inequality, we have $||AB||^2_F\leq ||A||^2_F||B||^2_F$.
%\end{fact}

\begin{comment}
\begin{align*}
  ||A-AR'{R'}^T||_F^2 &=Tr((A-AR'{R'}^T)^T(A-AR'{R'}^T))\\
  &=Tr((A^T-(AR'{R'}^T)^T)(A-AR'{R'}^T))\\
  &=Tr((A^T-R'{R'}^TA^T)(A-AR'{R'}^T))\\
  %&=Tr((A^T-R'{R'})^TA^T)(A-AR'{R'}^T))\\
   &=Tr(A^TA-A^TAR'{R'}^T-R'{R'}^TA^TA+R'{R'}^TA^TAR'{R'}^T)\\
   &=Tr(A^TA)-Tr(A^TAR'{R'}^T)-Tr(R'{R'}^TA^TA)+Tr(R'{R'}^TA^TAR'{R'}^T)\\
   &=Tr(A^TA)-Tr(A^TAR'{R'}^T)-Tr(A^TAR'{R'}^T)+Tr({R'}^TR'{R'}^T A^TAR')\\
   &=Tr(A^TA)-Tr(A^TAR'{R'}^T)-Tr(A^TAR'{R'}^T)+Tr(A^TAR'{R'}^T)\\
   &=Tr(A^TA)-Tr((AR')(AR')^T)\\%-Tr(A^TAR'{R'}^T)+Tr(A^TA)\\
   &=||A||_F^2-||AR'||_F^2\\
    \end{align*}
\end{comment}
 In the following, we state the definitions of  subspace and projective clustering.
 
\begin{defi}[Subspace clustering]\label{defi:jsubspace}
 Let $A\in \R^{n\times d}$ and $j$ be an integer less than $d$. Then, the problem of $j$-subspace clustering is 
 to find a $j$-dimensional subspace $L$ of $\R^d$  that minimizes the $\dist^2(A, L).$ 
 In other words, the goal is to find a matrix $X^{\perp}\in \R^{d\times (d-j)}$ having orthonormal columns 
 that minimizes $||AX^{\perp}||_F^2$ over every such possible matrix $X^{\perp}$.
\end{defi}
\begin{defi}[linear (affine) $(k, j)$-projective clustering]\label{defi:kjsubspace}
 Let $A\in \R^{n\times d}$,  $j$ be an integer less than $d$, and $k$ be an integer less than $n$. 
 Then, the problem of linear (affine) $(k, j)$-projective clustering is to find a 
 closed set $\mathcal{C}$,  which is the union of $k$ linear (affine) subspaces $\{L_1, \ldots L_k\}$ 
 each of dimension at most $j$, such that  it minimizes the $\dist^2(A, \mathcal{C})$,  over every possible choice of $\mathcal{C}$.
 \end{defi}
 \begin{comment}
\begin{defi}[Coreset for $\mathcal{C}$-clustering~\cite{PM04}]\label{defi:coreset}
Let $\mathcal{C}$ be a family of closed set in $\R^d$. Let $A\in \R^{n\times d}$, $r$ be an integer less than $n$, 
and $\epsilon>0$. We define $S\in \R^{r\times d}$ with a set of weights $\{w_i\}_{i=1}^r>0$ associated with 
its rows (let we denote the  $i$-th row of $\mathcal{S}$ by $S_{i*}$), is an $\epsilon$-coreset for 
$\mathcal{C}$-clustering if for every $C\in \mathcal{C}$, the following 
condition satisfies:
$$
(1-\epsilon)\dist^2(A, C)\leq \Sigma_{i=1}^r w_i \dist^2(S_{i*}, C) \leq (1+\epsilon)\dist^2(A, C).
$$

%Here, $S_{i*}$ denotes the $i$-th row of $S$.
\end{defi}
\end{comment}
\begin{theorem}[Low-rank approximation by \cite{sarlos06}]\label{theorem:randomProj}
 Let  $A \in \R^{n\times d}$,  and $\pi_{.}(.)$ denote the projection operators stated in the notation table. 
 If $\epsilon\in(0, 1]$ 
 and $\mathcal{S}$ is an $(r\times n)$ Johnson-Lindenstrauss matrix with i.i.d. zero-mean $\pm1$ entries and 
 $r=O\left(\frac{m}{\epsilon}\right)\log \frac{1}{\delta} $, then with probability 
 at least $1-\delta$ it holds that
  $$||A-\pi_{\mathcal{S}A, m}(A)||_F^2\leq (1+\epsilon)||A-A^{(m)}||_F^2.$$
Further, computing the singular vectors spanning $\pi_{\mathcal{S}A, m}(A)$ in
two passes 
\footnote{Two passes are required as we first multiply $A$ on the right with a Johnson-Lindenstrauss 
matrix $\mathcal{S}$,   and then we project the rows of $A$ 
again onto the row span of $\mathcal{S}A$.} 
over the data requires $O(\nnz(A)r+(n+d)r^2)$ time. %and $O((n + d)r^2)$ space, 
%where $\nnz(A)$ denotes the number of non-zeroes entries in $A$.
\end{theorem}

For our analysis, we will use a weak triangle inequality which is stated below:
\begin{lem}[Lemma $7.1$ of \cite{FSS13}]\label{lem:triangleIneq} For any $\varepsilon\in(0, 1)$, a closed set $\mathcal{C}$, and two points
$p, q\in \R^d$, we have
$$
 |\dist^2(p, \mathcal{C})-\dist^2(q, \mathcal{C})|\leq \frac{12||p-q||^2}{\varepsilon}+\frac{\varepsilon}{2}\dist^2(p, \mathcal{C}).
$$
\end{lem}

%\textcolor{blue}{State in footnote that why two passes are required in the above.}

\section{Faster coreset construction for subspace clustering}\label{sec:jSubspace}
In this section after revisiting the results of Cohen \textit{et al.}~\cite{CohenEMMP151},  
we present a randomized coreset construction for subspace clustering. The deterministic 
coreset construction of Feldman \textit{et al.}~\cite{FSS13} for subspace clustering 
problem relies on projecting the input matrix on its first few right singular vectors -- projecting 
the rows of $A$ on first few right singular vectors  of $A$ -- which requires SVD computation of $A$.
Cohen \textit{et al.}
~\cite{CohenEMMP151} suggested that projecting the rows of $A$ on 
some orthonormal vectors that closely approximate the right singular 
%on a few vectors which approximation of right singular 
vectors of $A$ (obtained via e.g.~\cite{sarlos06}) 
  also satisfies the required properties of coreset \textit{w.h.p.}, and as a consequence, gives a faster coreset construction. 
\begin{theorem}[Adapted from Theorem $8$ of~\cite{CohenEMMP151}]\label{theorem:jSubspaceCoreset}
Let  $X  \in \R^{d\times j}$ be an orthonormal matrix representing a subspace $L$, let
  $X^{\perp} \in \R^{d\times (d-j)}$ be the orthonormal matrix  representing the orthogonal complement of
 $L$, $\epsilon \in (0,1)$, $\delta \in (0,1)$, $m=\lceil\frac{j}{\epsilon}\rceil$, $\Delta=||A-A^{(m)}||_F^2$, and $\tilde{A}$ is 
a rank $m$ approximation of $A$ satisfying Theorem~\ref{theorem:randomProj}. Then, the following is true 
with probability at least $1-\delta$:
     $$0\leq \left|||\tilde{A}X^{\perp}||_F^2+\Delta-||AX^{\perp}||_F^2\right| \leq 2\epsilon||AX^{\perp}||_F^2.$$
%Expected time bound of Algorithm~\ref{algorithm:algoJsubspace} is $O(nnz(A)( \frac{m}{\epsilon}+ m \log m )+(n + d)(\frac{m}{\epsilon}+m\log m)^2+ndm).$
\end{theorem}

\begin{proof}
 
 Using a result of  Sarl{\'{o}}s~\cite{sarlos06}, we get a rank $m$ approximation of $A$. If 
$\mathcal{S}$ is an $(r\times n)$ JL matrix, where $r=O\left((\frac{m}{\epsilon}+m\log m)\log \frac{1}{\delta} \right)$
(see Theorem~\ref{theorem:randomProj})
then the following is true with probability at least $1-\delta$:
 \begin{equation}\label{eq:sarlos06}
 ~~~~~~~~~~~~||A-\pi_{\mathcal{S}A, m}(A)||_F^2\leq (1+\epsilon)||A-A^{(m)}||_F^2.
\end{equation}
Here, $A^{(m)}$ is the best $m$ rank approximation of $A$. 
%If $V'$ is the matrix having the first $m$ right singular vectors of $A$, then $A^{(m)}=A^{(m)}$. 
Let $R'$ be the matrix having the first $m$ right singular vectors 
of $\pi_{\mathcal{S}A}(A)$, and let we denote  $AR'R'^T$ by $\tilde{A}$,  
then by Equation~\ref{eq:sarlos06}, the following holds true with probability at least $1-\delta$:
 \begin{align*}
 ||A-\tilde{A}||^2_F &\leq (1+\epsilon) ||A-A^{(m)}||^2_F \numberthis\label{eq:randomProjEq}
 \end{align*}

In the following we  show an upper bound on the following expression:
\begin{align*}
& \left|||\tilde{A}X^{\perp}||_F^2+\Delta-||AX^{\perp}||_F^2\right|\\
&=\left|||\tilde{A}||_F^2-||\tilde{A}X||_F^2+||A-A^{(m)}||_F^2-||A||_F^2+||AX||_F^2\right| \numberthis\label{eq:eq24}\\%~~~~~~~~~~~~~~~~~~~~~~~~~~~~~~~~~~~~\mbox{By pythagoras theorem}\\
&=\left|||\tilde{A}||_F^2-||\tilde{A}X||_F^2+||A||_F^2-||A^{(m)}||_F^2-||A||_F^2+||AX||_F^2\right|\numberthis\label{eq:eq924}\\
&=\left|||\tilde{A}||_F^2-||A^{(m)}||_F^2-||\tilde{A}X||_F^2+||AX||_F^2\right|\\
&\leq \left|||A^{(m)}||_F^2-||A^{(m)}||_F^2+||AX||_F^2- ||\tilde{A}X||_F^2\right|\numberthis\label{eq:eq1094}\\
&= \left|||AX||_F^2- ||\tilde{A}X||_F^2\right|
\leq 2\epsilon||AX^{\perp}||_F^2 \numberthis\label{eq:upperboundJ}
 \end{align*}
   Equality~\ref{eq:eq24} follows from Pythagorean theorem;  
   Equality~\ref{eq:eq924} follows from Fact~\ref{fact:fact2}, where $A^{(m)}=AV'V'^T$, and $V'\in\R^{d\times m}$ 
  having $m$ columns from the first $m$ right singular vectors of $A$;
  Inequality~\ref{eq:eq1094} holds as the value of  $||\tilde{A}||_F^2-||A^{(m)}||_F^2$ can be at most zero, because 
    at the best we can hope to sample the right singular vectors of $A$ as $R'$, which maximizes the value of the desired expression; 
  finally Inequality~\ref{eq:upperboundJ} holds  
   from Lemma~\ref{lem:projbound}. 
 \end{proof}
 A proof of the following lemma follows from the analysis of Lemma 5 of~\cite{CohenEMMP151}.    
 %  We defer it   to the full version of this paper.
 We defer its proof in the appendix. % due to space limit.
%The remaining proof %is presented in the appendix, where we show that for $\epsilon\in(0, \epsilon')$ with $\epsilon'<\frac{3j}{3d-4j+1}$, we have the following:
\begin{lem}[Adapted from Lemma $5$ and Theorem $8$ of~\cite{CohenEMMP151}]\label{lem:projbound}
 Let $A\in \R^{n\times d}$,  $\tilde{A}$ is a rank $m$ approximation of $A$ satisfying Equation~\ref{eq:randomProjEq}, then
 %let    $X^{\perp} \in \R^{d\times (d-j)}$ be an orthonormal matrix, then 
 $$
  0\leq ||AX||_F^2- ||\tilde{A}X||_F^2 \leq 2\epsilon||AX^{\perp}||_F^2.
$$
\end{lem}
\section{Faster coreset construction for projective clustering}\label{sec:kjSubspace}
In this section, extending the result (Theorem~\ref{theorem:jSubspaceCoreset}) of the previous section, we present  
 a randomized coreset construction for the problem of projective clustering. 
More precisely, if  $L_1,..., L_k$ be a set of $k$ subspaces each of dimension at most $j$, and let $\mathcal{C}$ 
 be a closed set containing union of them, then our  randomized coreset is a matrix of very small rank 
 (independent of $d$) and it approximately preserves the distances  from every such closed set $\mathcal{C}$, 
 with high probability. Our main contribution is the dimensionality  
 reduction step of  the  coreset construction, which is presented in Algorithm $1$ below.\\%~\ref{algorithm:AlgodimredKJsubspace}.\\
\vspace{-1cm}
\begin{algorithm}\label{algorithm:AlgodimredKJsubspace}
%\SetAlgoLined
~\KwInput{$A\in \R^{n\times d}$, an integer $1\leq j<d-1$, and an integer $1\leq k<n-1$ such that 
$ j^*\leq d-1$,   % with $\epsilon'=\min\left\{\frac{2}{3j^*}, \sqrt{\frac{26j^*}{d-j^*-0.6}}\right\},$
where $j^*=k(j+1)$,   $\epsilon\in(0, 1)$, $\delta \in (0, 1)$.}
~\KwResult{Dimensionality reduction for randomized coreset construction for the projective clustering.}
~Compute  an  Johnson-Lindenstrauss matrix $[\mathcal{S}]_{r\times n}$ having \textit{i.i.d.} $\pm1$ entries and zero-mean,
where $r=O(\frac{m^*}{\epsilon})\log \frac{1}{\delta}$, 
$m^*=\lceil\frac{52j^*}{\epsilon^2}\rceil$.\\
~Compute the matrix $\pi_{\mathcal{S}A}(A)$.\\

~Compute the SVD of  $\pi_{\mathcal{S}A}(A)$, let $R^*\in \R^{d\times m^* }$ be 
the first $m^*$ right singular vectors of $\pi_{\mathcal{S}A}(A)$.\\

~Let us denote $AR^*{R^*}^T$ by $A^*$, and output $A^*$.
  \caption{Dimensionality reduction for  projective clustering.}
\end{algorithm}\\
%In the following, 
%we present a proof of Theorem~\ref{thm:thmmainKJ} that discuss the bounds achieved by the Algorithm~\ref{algorithm:AlgodimredKJsubspace}.
\textbf{Proof of Theorem~\ref{thm:thmmainKJ}: }
Let $[{X^*}]_{d\times j^*}$ be a matrix with orthonormal columns whose span is $L^*$, 
 and let ${L^*}^{\perp}$ be the orthogonal complement of $L^*$ spanned by  $[{{X^*}^{\perp}}]_{d\times (d-j^*)}$. 
 If $\mathcal{C}$ is a closed set spanned by $L^*$, then due to the Pythagorean theorem, we have, 
%In our analysis, we   use the equality, 
 $\dist^2(A, \mathcal{C})=||A{{X^*}^{\perp}}||^2_F+\dist^2(A{X^*}{X^*}^T, \mathcal{C})$. Further, 
 \begin{align*}
 &\left|\left(\dist^2(A^*, \mathcal{C})+\Delta^*\right)-\dist^2(A, \mathcal{C})\right|\\
% &= \left|\left(||A^*{{X^*}^{\perp}}||^2_F+\dist^2(A^*{X^*}{X^*}^T, \mathcal{C}) +\Delta^*\right)-\left(||A{{X^*}^{\perp}}||^2_F+\dist^2(A{X^*}{X^*}^T, \mathcal{C})\right)\right|\\
 &= \left|\left(||A^*{{X^*}^{\perp}}||^2_F+\dist^2(A^*{X^*}{X^*}^T, \mathcal{C}) +\Delta^*\right)-\left(||A{{X^*}^{\perp}}||^2_F+\dist^2(A{X^*}{X^*}^T, \mathcal{C})\right)\right|\\
%&\leq\underbrace{\left|\left(||A^*{{X^*}^{\perp}}||^2_F+\Delta^*-||A{{X^*}^{\perp}}||^2_F\right)\right|}_\text{first term }+\underbrace{\left|\left(\dist^2(A^*{X^*}{X^*}^T, \mathcal{C})-\dist^2(A{X^*}{X^*}^T, \mathcal{C}))\right)\right|}_\text{second term}\nonumber
&\leq\underbrace{\left|\left(||A^*{{X^*}^{\perp}}||^2_F+\Delta^*-||A{{X^*}^{\perp}}||^2_F\right)\right|}_\text{first term }+\underbrace{\left|\left(\dist^2(A^*{X^*}{X^*}^T, \mathcal{C})-\dist^2(A{X^*}{X^*}^T, \mathcal{C}))\right)\right|}_\text{second term}\nonumber
\end{align*}
     We have to bound two terms  in the above expression. The first term can be upper bounded using a similar 
    analysis as of Theorem~\ref{theorem:jSubspaceCoreset}
    which holds     true with probability at least $1-\delta$. 
    (In Theorem~\ref{theorem:jSubspaceCoreset}, we replace $j$ by $j^*, m$ by $m^*, \epsilon$ by $\frac{\epsilon^2}{52}$, 
    and $\Delta$ by $\Delta^*$.)
    \begin{equation}\label{eq:equation106}
   \left|||A^*{{X^*}^{\perp}}||^2_F+\Delta^*-||A{{X^*}^{\perp}}||^2_F\right|\leq  \frac{\epsilon^2}{26}||A{{X^*}^{\perp}}||^2_F
 \end{equation}
To   bound the second term $\left|\dist^2(A^*{X^*}{X^*}^T, \mathcal{C})-\dist^2(A{X^*}{X^*}^T, \mathcal{C}))\right|$,
 we   use   Lemma~\ref{lem:triangleIneq}. 
For any $\varepsilon\in (0, 1)$ and due to Lemma~\ref{lem:triangleIneq}, we have
\begin{align*}
&\left|\dist^2(A^*{X^*}{X^*}^T, \mathcal{C})-\dist^2(A{X^*}{X^*}^T, \mathcal{C}))\right|\\
 &\leq \frac{12}{\varepsilon}||A^*{X^*}{X^*}^T-A{X^*}{X^*}^T||^2_F+\frac{\varepsilon}{2}\dist^2(A{X^*}{X^*}^T, \mathcal{C})\\
&\leq \frac{12}{\varepsilon}\left( \frac{\epsilon^2}{26}||A{X^*}^{\perp}||^2_F\right)+\frac{\varepsilon}{2}\dist^2(A{X^*}{X^*}^T, \mathcal{C})\numberthis\label{eq:eq500}\\
&\leq \frac{12}{\varepsilon}\left( \frac{\epsilon^2}{26}||A{X^*}^{\perp}||^2_F \right)+\frac{\varepsilon}{2}\dist^2(A, \mathcal{C})\nonumber
\end{align*}
Inequality~\ref{eq:eq500} holds  due to Lemma~\ref{lem:difFrobNorm}.
Thus, we have
\begin{align*}%\label{eq:equation107}
&\left|\dist^2(A^*{X^*}{X^*}^T, \mathcal{C})-\dist^2(A{X^*}{X^*}^T, \mathcal{C}))\right|
\leq \frac{12}{\varepsilon}\left( \frac{\epsilon^2}{26}||A{X^*}^{\perp}||^2_F\right)+\frac{\varepsilon}{2}\dist^2(A, \mathcal{C})\numberthis\label{eq:equation107}
\end{align*}

%\begin{align*}%\label{eq:equation107}

 Equation~\ref{eq:equation106},  
    in conjunction with Equation~\ref{eq:equation107}, gives us the following:
\begin{align*}
&\left|\left(\dist^2(A^*, \mathcal{C})+\Delta^*\right)-\dist^2(A, \mathcal{C})\right|\\
&\leq \left(1+\frac{12}{\varepsilon}\right)\frac{\epsilon^2}{26}||A{{X^*}^{\perp}}||^2_F+\frac{\varepsilon}{2}\dist^2(A, \mathcal{C})\\
&\leq\left(1+\frac{12}{\varepsilon}\right)\frac{\epsilon^2}{26}\dist^2(A, \mathcal{C})+\frac{\varepsilon}{2}\dist^2(A, \mathcal{C})\\
&=\left(\frac{\epsilon^2}{26}+\frac{12\epsilon^2}{26\varepsilon}+\frac{\varepsilon}{2}\right)\dist^2(A, \mathcal{C})\\
&=\left(\frac{\epsilon^2}{26}+\frac{12\epsilon}{26}+\frac{\epsilon}{2}\right)\dist^2(A, \mathcal{C})\numberthis\label{eq:eq5291}\\
%&=\left(\frac{\epsilon^2}{26}+\frac{12\epsilon}{26}+\frac{\epsilon}{2}\right)\dist^2(A, \mathcal{C})\\
%&\leq\left(\frac{\epsilon^2}{26}+\frac{12\epsilon}{26}+\frac{\epsilon}{2}\right)\dist^2(A, \mathcal{C})\\
&\leq \epsilon \dist^2(A, \mathcal{C})
\end{align*}
Equality~\ref{eq:eq5291} holds by choosing $\varepsilon=\epsilon$, and as $\epsilon^2/26+12\epsilon/26<\epsilon/2$.

%%%%%% Proof ends here%%%%%%%%%%%%%%%%%%%%%%%%%
%We defered its proof  to the full version of this paper.
 A proof of the following lemma is presented in the appendix. % to the full version of this paper.
\begin{lem}\label{lem:difFrobNorm}
 Let ${X^*} \in \R^{d\times j^*}$ be a matrix with orthonormal 
columns whose span is $L^*$, then in Algorithm $1$ %~\ref{algorithm:AlgodimredKJsubspace} 
the following is true with probability at least $1-\delta$
$$
  ||A^*{X^*}{X^*}^T-A{X^*}{X^*}^T||^2_F\leq 
\frac{\epsilon^2}{26}||A{{X^*}}^{\perp}||^2_F.
$$
 \end{lem}

%\begin{proof} \begin{align*}
%  ||A^*{X^*}{X^*}^T-A{X^*}{X^*}^T||^2_F  &=||A{X^*}{X^*}^T-A^*{X^*}{X^*}^T||^2_F\\
%  &=||(A-A^*){X^*}{X^*}^T||^2_F\\
% &=\tr\left((A-A^*){X^*}{X^*}^T\left((A-A^*){X^*}{X^*}^T  \right)^T\right)\\
% &=\tr\left((A-A^*){X^*}{X^*}^T{X^*}{X^*}^T(A-A^*)^T\right)\\
% &=\tr\left((A-A^*)^T(A-A^*){X^*}{X^*}^T\right)\numberthis\label{eq:eq1232}\\
% &=\tr\left({X^*}{X^*}^T(A^TA-{A^*}^T{A^*})\right)\numberthis\label{eq:eq1233}\\
% &\leq\frac{\epsilon^2}{26}||A{{X^*}}^{\perp}||^2_F\numberthis\label{eq:projlem11}~~~
% \end{align*}
% Equality~\ref{eq:eq1232} holds due to cyclic property of $\tr$ function and ${X^*}^T{X^*}=I$;
% Equality~\ref{eq:eq1233} holds due to cyclic property of $\tr$, and 
% $(A-A^*)^T(A-A^*)=(A-AR^*{R^*}^T)^T(A-AR^*{R^*}^T)=A^TA-{A^*}^TA^*$ (after simplification, see Fact~\ref{fact:fact2}); 
% finally  Inequality~\ref{eq:projlem11} holds by following the steps of proof of 
% Lemma~\ref{lem:projbound} (from Equation~\ref{eq:eq999}), 
% where we replace  $\epsilon$ by $\frac{\epsilon^2}{52};$   $j$ by $j^*$, $A$ by $A^*$, $X$ by $X^*$.
% \end{proof}

  \begin{remark}
   Please note that it is  sufficient to store the matrix $AR^*$ which is of 
   dimension $m^*$, where $m^*=O\left({k(j+1)}/{\epsilon^2}\right)$. However, for the
purpose of our analysis, we use the matrix $AR^*{R^*}^T$ 
   which is of dimension $d$, and rank $m^*$. Further, the space  that is required to store our
   coreset is $O(nm^*+1)$ -- we need $O(nm^*)$ space to store the matrix $AR^*$, and 
   $O(1)$ space to store the term $\Delta^*$; on the other hand, 
the space require to store $A$ is $O(nd).$
  \end{remark}

\noindent\textbf{Comparison with coreset construction of~\cite{FSS13}:}
 %We will compare and contrast our result with that of~\cite{FSS13}. %Coreset defined in \cite{FSS13} 
  Coreset construction of~\cite{FSS13}  requires projecting the rows of $A$ on its first $O(k(j+1)/\epsilon^2)$ 
  right singular vectors which gives 
 a matrix of rank  $O(k(j+1)/\epsilon^2)$ and it approximately preserves the 
distance from any closed $\mathcal{C}$. Their construction requires computing SVD of the given matrix $A$, 
 which has the run-time complexity of
 $\min\{n^2d, nd^2\}$. 
In our construction, we  showed that it is also sufficient  to 
project the rows of $A$ on $O(k(j+1)/\epsilon^2)$ othronormal vectors 
that closely approximate the right singular vectors of $A$. 
%As mentioned earlier, we obtain those othronormal vectors  using~\cite{sarlos06}, and consequently achieve  a faster algorithm when the values of $k$ and $j$ are small.
  We now give  an   time bound on the running time of Algorithm $1$. %~\ref{algorithm:AlgodimredKJsubspace}. 
 Time required for execution of line number $3, 4, 5$ is 
 \begin{align*}
%  &O\left(\nnz(A)( \frac{m^*}{\epsilon}+ m^* \log m^* )+(n + d)(\frac{m^*}{\epsilon}+m^*\log m^*)^2\right) \\
%  &=O\left(\nnz(A)( \frac{j^*}{\epsilon^3}+ \frac{j^*}{\epsilon^2} \log \frac{j^*}{\epsilon^2} )+(n + d)(\frac{j^*}{\epsilon^3}+\frac{j^*}{\epsilon^2}\log \frac{j^*}{\epsilon^2})^2\right),
 &O\left(\nnz(A) \frac{m^*}{\epsilon} +(n + d)\left(\frac{m^*}{\epsilon}\right)^2\right) 
 =O\left(\nnz(A) \frac{j^*}{\epsilon^3} +(n + d) \frac{{j^*}^2}{\epsilon^6} \right),
\end{align*}
 % $O(\nnz(A)( \frac{j^*}{\epsilon^3}+ \left(\frac{j^*}{\epsilon^2}\right) \log \left(\frac{j^*}{\epsilon^2} \right))+(n + d)\left(\frac{j^*}{\epsilon^3}+\frac{j^*}{\epsilon^2}\log \frac{j^*}{\epsilon^2}\right)^2\right)$ 
 due to~\cite{sarlos06}, where $j^*=k(j+1)$. Further, line number $6$ 
 requires time - for projecting $A$ on $R^*$, which due to an elementary matrix multiplication 
 is $O(\nnz(A)m^*)=O(\nnz(A)\frac{j^*}{\epsilon^2})$. %=O\left(\frac{ndj^*}{\epsilon^2}\right)$. 
  Thus, total   running time of  
 Algorithm $1$ %~\ref{algorithm:AlgodimredKJsubspace} 
 is 
 \begin{align*}
 &O\left(\nnz(A) \frac{j^*}{\epsilon^3} +(n + d) \frac{{j^*}^2}{\epsilon^6} +\nnz(A)\frac{j^*}{\epsilon^2}\right)
 =&O\left(\nnz(A) \frac{j^*}{\epsilon^3} +(n + d) \frac{{j^*}^2}{\epsilon^6}\right).
%  &O\left(\nnz(A)\left( \frac{j^*}{\epsilon^3}+ \frac{j^*}{\epsilon^2} \log \frac{j^*}{\epsilon^2} \right)+(n + d)\left(\frac{j^*}{\epsilon^3}+\frac{j^*}{\epsilon^2}\log \frac{j^*}{\epsilon^2}\right)^2+\frac{ndj^*}{\epsilon^2}\right)\\
%   &=\tilde{O}\left(\nnz(A) \frac{j^*}{\epsilon^3} + (n + d)\frac{{j^*}^2}{\epsilon^6}+\frac{ndj^*}{\epsilon^2}\right).
%  
\end{align*}
 Clearly, if $n\geq d$ and $j^*=o(n)$, or, if  $n<d$ and $j^*=o(d)$, then
our   running time is better than that of~\cite{FSS13}.

% \begin{comment}
  As a corollary of Theorem~\ref{thm:thmmainKJ}, and using the known techniques 
from~\cite{FSS13,VaradarajanX12,FeldmanL11} on $A^*$, we present the cardinality reduction 
step of coreset construction as follows:

\begin{cor}[Corollary $9.1$ of \cite{FSS13}]
 Let $A\in \{1, 2,\ldots,\Lambda\}^{n\times d}$, with $\Lambda\in (nd)^{O(1)}$,  $d\in n^{O(1)}.$
  There is a matrix $\mathcal{Q}\in \R^{l\times d'}$ with $l=\poly(2^{kj},\frac{1}{\epsilon}, \log n, \log \Lambda)$, 
  $d'=O({k(j+1)}/{\epsilon^2})$; and a weight function associated with the rows of $\mathcal{Q}$, 
  \textit{i.e.} $w:\mathcal{Q}_{i^*}\rightarrow[0, \infty)$ such that 
  for every closed set $\mathcal{C}$, which is the union of $k$ affine $j$-subspaces of $\R^d$, 
  the following holds with high  probability
 \[
  (1-\epsilon)\Sigma_{i=1}^n\dist^2(A_{i^*}, \mathcal{C}) \leq  \Sigma_{i=1}^{l}w(\mathcal{Q}_{i^*})\dist^2(\mathcal{Q}_{i^*}, \mathcal{C})
  \leq(1+\epsilon)\Sigma_{i=1}^n\dist^2(A_{i^*}, \mathcal{C}).
 \]
\end{cor}
% \end{comment}
 As a corollary of Theorem~\ref{thm:thmmainKJ}, we present randomized coreset  result 
 for $k$-mean clustering. However, Theorem $8$ of~\cite{CohenEMMP151} independently offers a tighter dimension reduction bound  $\lceil \frac{k}{\epsilon} \rceil$ for $k$-means.
  \begin{cor}
    Let $A\in \R^{n\times d}$, $\epsilon\in(0, 1)$, and $k$ an integer less than $(d-1)$ and $(n-1)$. 
  Then there is a randomized algorithm which outputs a matrix $A'$ of dimension $O({k}/{\epsilon^2})$ %$O\left(\frac{k}{\epsilon^2}\right)$
    such that for every set of $k$ points $\{c_i\}_{i=1}^k\in \R^d$ represented as the rows of matrix $C$,  
    the following holds with high probability:
 $$
   |\left(\dist^2(A', C)+\Delta'\right)-\dist^2(A, C)|\leq \epsilon \dist^2(A, C).
$$
The   running time of the algorithm is
$ {O}\left(\nnz(A) \frac{k}{\epsilon^3} + (n + d)\frac{{k}^2}{\epsilon^6}\right).$
     %$O(\nnz(A)( \frac{m'}{\epsilon}+ m' \log m' )+ (n + d)(\frac{m'}{\epsilon}+m'\log m')^2+ndm')$.
     %$O(\nnz(A)( \frac{k}{\epsilon^3}+ \frac{k}{\epsilon^2} \log \frac{k}{\epsilon^2} )+(n + d)(\frac{k}{\epsilon^3}+ \frac{k}{\epsilon^2}\log \frac{k}{\epsilon^2})^2+\frac{ndk}{\epsilon^2}).$
   Where, %$m'=O({k}/{\epsilon^2})$; 
   $\Delta'= ||A-A^{O({k}/{\epsilon^2})}||^2_F$; and $\dist^2(A, C)$ denotes the sum 
   of square distances from each row of $A$ to its closest point in $C$. 
  \end{cor}

\vspace{-0.5cm}
\section{Conclusion and open problems}\label{sec:conclusion} 
 We presented a randomized coreset construction for projective clustering 
 \textit{via} low rank approximation. We first revisited the result of~\cite{CohenEMMP151} 
 for the subspace clustering, and then  extended their result to construct a randomized coreset  for projective clustering. 
 We showed that our construction is significantly faster (when the values of $k$ and $j$ are small), 
 as compared to the corresponding  deterministic construction of~\cite{FSS13}, and it also maintains nearly 
the same accuracy. Our work leaves  several open problems - improving the dimensionality reduction bounds for projective
 clustering, or giving a matching lower bound for the same. Another important open problem is to come up with the 
 dimension reduction step of coreset construction using feature selection algorithms such as \textit{row/column 
 subset selection}~\cite{BoutsidisMD09}.

\bibliographystyle{abbrv}
\bibliography{reference}
\appendix
 \section{Appendix}
%\textbf{Proof of Fact~\ref{fact:fact2}: }

\textbf{Proof of Lemma~\ref{lem:projbound}: }
%Due to cyclic and linear properties of $\tr$ function it is easy to verify that 
% \begin{align*}
%||AX||_F^2- ||\tilde{A}X||_F^2&=\tr(XX^T(A^TA-\tilde{A}^T\tilde{A}))\numberthis\label{eq:eq999}.
% \end{align*}
We first express the term $||AX||_F^2- ||\tilde{A}X||_F^2$ in the form of $\tr$ function:
\begin{align*}
  &||AX||_F^2- ||\tilde{A}X||_F^2=\tr((AX)^T(AX))-\tr((\tilde{A}X)^T\tilde{A}X)\\
  &=\tr(X^TA^TAX)-\tr(X^T\tilde{A}^T\tilde{A}X)\\
  &=\tr(X^T(A^TA-\tilde{A}^T\tilde{A})X)=\tr(XX^T(A^TA-\tilde{A}^T\tilde{A}))\numberthis\label{eq:eq999}.
 \end{align*}
 The above equalities follows due to definition of $\tr$ function   $||A||^2_F=\tr(A^TA)$, and due to cyclic 
 and linear properties of $\tr$ function (Fact~\ref{fact:trace}).  
Let we denote the matrix $XX^T$ by $P$, and $(A^TA-\tilde{A}^T\tilde{A})$ by matrix $M$. 
Thus, the problem reduces to bounding the term $\tr(PM)$. 
Let $\lambda_i(M)$ is the $i$th eigenvalue, and 
$\{w_i\}^d_{i=1}$ be the eigenvectors of $M$, then 
$M=\Sigma_{i=1}^d \lambda_i(M)w_iw_i^T$.%, where $\sigma_i(M)$ is the $i$th eigenvalue of  $M.$ Thus,
The following expression holds due to linearity of trace function.
\begin{align*}
 \tr(PM)=\tr\left(P \Sigma_{i=1}^d \lambda_i(M)w_iw_i^T \right)= \Sigma_{i=1}^d \lambda_i(M)\tr\left(Pw_iw_i^T \right)
 \end{align*}
Further, we bound the summation $\Sigma_{i=1}^d\tr\left(Pw_iw_i^T \right)$,
\begin{align*}
 \Sigma_{i=1}^d\tr\left(Pw_iw_i^T \right)
 &=\tr(PWW^T)=\tr(P^TPWW^TWW^T)
 =\tr(P^TPWW^T)\\
 &=\tr(PWW^TP^T)=||PW||^2_F\leq ||P||_F^2=||XX^T||_F^2\leq j \numberthis\label{eq:eq1200}
\end{align*}
where, $W \in \R^{d\times d}$ having columns as eigenvectors of $M.$ 
The above euqalities follow as $P=XX^T$, then $P^TP=P$; similarly $WW^T=WW^TWW^T$, also $X^TX=I, W^TW=I.$
Finally, the inequality $||PW||^2_F\leq ||P||_F^2$ follows from Fact~\ref{fact:projDecNorm}.

Further,  $P=XX^T$ has all singular values either $1$ or $0$. 
\begin{align*}
 0\leq \tr\left(Pw_iw_i^T \right)=w_i^TPw_i\leq ||w_i||^2_2||P||^2_2\leq 1 \numberthis\label{eq:eq1201}
\end{align*}
Thus, for  $1\leq i\leq d$,  $\tr\left(Pw_iw_i^T \right)$ has $d$ values, 
and each value is at most $1$ (Equation~\ref{eq:eq1201}), and 
sum of all of them is at most $j$ (Equation~\ref{eq:eq1200}). 

Further as $M$ is symmetric,    we have 
\begin{align*}
M&=A^TA-{\tilde{A}}^T\tilde{A}= A^TA-(A-(A-\tilde{A}))^T(A-(A-\tilde{A}))\\
&= A^TA-(A^T-(A-\tilde{A})^T)(A-(A-\tilde{A}))\\
&= A^TA- A^TA + (A-\tilde{A})^TA + A^T(A-\tilde{A})-(A-\tilde{A})^T(A-\tilde{A})\\
&=   (A-\tilde{A})^TA + A^T(A-\tilde{A})-(A-\tilde{A})^T(A-\tilde{A})\\
&=   (A-\tilde{A})^T(A-\tilde{A}+\tilde{A}) + (A-\tilde{A}+\tilde{A})^T(A-\tilde{A})-(A-\tilde{A})^T(A-\tilde{A})\\
&=   (A-\tilde{A})^T(A-\tilde{A})+(A-\tilde{A})^T\tilde{A} + (A-\tilde{A})^T(A-\tilde{A}) + \tilde{A}^T(A-\tilde{A})\ldots \\&\ldots-(A-\tilde{A})^T(A-\tilde{A})\\
&=(A-\tilde{A})^T(A-\tilde{A})\numberthis\label{eq:eq1802}
%&=\left(A-AR'{R'}^T\right)^T\left(A-AR'{R'}^T\right)\numberthis\label{eq:eq1802}
\end{align*}
%$M=A^TA-{\tilde{A}}^T\tilde{A}=||A||^2_F-||AR'{R'}^T||^2_F=||A-AR'{R'}^T||^2_F=\left(A-AR'{R'}^T\right)^T\left(A-AR'{R'}^T\right)$ 
%Equality~\ref{eq:eq1802} holds due to Fact~\ref{fact:fact2}.   
Equality~\ref{eq:eq1802} holds due to $(A-\tilde{A})^T\tilde{A}=\tilde{A}^T(A-\tilde{A})=0$ because rows of $\tilde{A}$ and 
$(A-\tilde{A})$ lie in orthogonal subspaces.
Equality~\ref{eq:eq1802} and Fact~\ref{fact:posSemiDef} 
shows that  $M$ is a positive semidefinite matrix,  and as a consequence it
has all nonnegative eigenvalues. Then, the summation $\Sigma_{i=1}^d\lambda_i(M)\tr\left(Pw_iw_i^T \right)$ is maximized 
when  $\tr\left(Pw_iw_i^T \right)=1$, that is for those eigenvectors
  which corresponds to  largest magnitude eigenvalues of $M$. Thus,
\[
 0\leq \tr(PM)=\Sigma_{i=1}^d \lambda_i(M)\tr\left(Pw_iw_i^T \right)\leq \Sigma_{i=1}^j \lambda_i(M).
\]
As a consequence, we have $0\leq \tr(XX^T(A^TA-\tilde{A}^T\tilde{A}))\leq \Sigma_{i=1}^j \lambda_i(A^TA-\tilde{A}^T\tilde{A})=\Sigma_{i=1}^j \sigma^2_i(A-\tilde{A})=
||(A-\tilde{A})_j||^2_F.$ Here matrix $(A-\tilde{A})_j$ is the matrix restricted to rank $j$ of the matrix $A-\tilde{A}$. Further, 
as matrix $\tilde{A}$ is of rank at most $m$, $\tilde{A}+(A-\tilde{A})_j$ is of rank at most $m+j.$
Thus, we have 
\begin{align*}
 ||A-\left(\tilde{A}+(A-\tilde{A})_j\right)||^2_F &\geq \Sigma^d_{i=m+j+1} \sigma_i^2 \numberthis\label{eq:eq1202}\\
 ||A-\tilde{A}||_F^2-||(A-\tilde{A})_j||^2_F &\geq \Sigma^d_{i=m+j+1} \sigma_i^2 \numberthis\label{eq:eq1203}\\
 ||(A-\tilde{A})_j||^2_F &\leq ||A-\tilde{A}||_F^2-\Sigma^d_{i=m+j+1} \sigma_i^2 \\
 ||(A-\tilde{A})_j||^2_F &\leq (1+\epsilon)||A-A^{(m)}||_F^2-\Sigma^d_{i=m+j+1} \sigma_i^2 \numberthis\label{eq:eq1205}\\
 &= (1+\epsilon)\Sigma^d_{i=m+1}\sigma_i^2-\Sigma^d_{i=m+j+1} \sigma_i^2\\
 &= \Sigma^{m+j}_{i=m+1}\sigma_i^2 +\epsilon\Sigma^d_{i=m+1} \sigma_i^2 \\
 &\leq \epsilon \Sigma_{j+1}^{d}\sigma_i^2 +\epsilon\Sigma^d_{i=j+1} \sigma_i^2\numberthis\label{eq:eq1206}\\
 &=2\epsilon||AX^{\perp}||_F^2
\end{align*}
Inequality~\ref{eq:eq1202} follows from Fact~\ref{fact:forbNormDef};
Inequality~\ref{eq:eq1203} follows from Fact~\ref{fact:fact2}, where 
$(A-\tilde{A})_j=(A-\tilde{A})XX^T$ for an orthonormal matrix $X\in \R^{d\times j}$;
Inequality~\ref{eq:eq1205} follows from Theorem~\ref{theorem:randomProj};
Inequality~\ref{eq:eq1206} holds as $m=\lceil\frac{j}{\epsilon}  \rceil$, and due to the following:
\[
 \Sigma^{m+j}_{i=m+1}\sigma_i^2\leq j\sigma_{m+1}^2= \epsilon m \sigma_{m+1}^2\leq \epsilon m \sigma_{j+1}^2 \leq 
 \epsilon \Sigma^{m+j}_{i=j+1}\sigma_i^2 \leq \epsilon \Sigma^{d}_{i=j+1}\sigma_i^2.
\]

\textbf{Proof of Lemma~\ref{lem:difFrobNorm}} \begin{align*}
  ||A^*{X^*}{X^*}^T-A{X^*}{X^*}^T||^2_F  &=||A{X^*}{X^*}^T-A^*{X^*}{X^*}^T||^2_F\\
  &=||(A-A^*){X^*}{X^*}^T||^2_F\\
 &=\tr\left((A-A^*){X^*}{X^*}^T\left((A-A^*){X^*}{X^*}^T  \right)^T\right)\\
 &=\tr\left((A-A^*){X^*}{X^*}^T{X^*}{X^*}^T(A-A^*)^T\right)\\
 &=\tr\left((A-A^*)^T(A-A^*){X^*}{X^*}^T\right)\numberthis\label{eq:eq1232}\\
 &=\tr\left({X^*}{X^*}^T(A^TA-{A^*}^T{A^*})\right)\numberthis\label{eq:eq1233}\\
 &\leq\frac{\epsilon^2}{26}||A{{X^*}}^{\perp}||^2_F\numberthis\label{eq:projlem11}~~~
 \end{align*}
 Equality~\ref{eq:eq1232} holds due to cyclic property of $\tr$ function and ${X^*}^T{X^*}=I$;
 Equality~\ref{eq:eq1233} holds due to cyclic property of $\tr$, and 
 $(A-A^*)^T(A-A^*)=(A-AR^*{R^*}^T)^T(A-AR^*{R^*}^T)=A^TA-{A^*}^TA^*$ (after simplification, similar to Equation~\ref{eq:eq1802}); % see Fact~\ref{fact:fact2}); 
 finally  Inequality~\ref{eq:projlem11} holds by following the steps of proof of 
 Lemma~\ref{lem:projbound} (from Equation~\ref{eq:eq999}), 
 where we replace  $\epsilon$ by $\frac{\epsilon^2}{52};$   $j$ by $j^*$, $A$ by $A^*$, $X$ by $X^*$.

\end{document}